\documentclass[11pt,a4paper]{article}
\usepackage{amsmath,amsfonts,amssymb,amsthm,latexsym}
\usepackage{graphics,epsfig,color,enumerate}
\usepackage{setspace,mathtools,enumitem,dsfont,verbatim}
\usepackage{caption}
\usepackage{subcaption}
\usepackage[english]{babel}
\makeatletter
\def\@captype{figure}
\makeatother

\newtheorem{theorem}{Theorem}[section]
\newtheorem{formula}{Formula}[section]
\newtheorem{lemma}[theorem]{Lemma}

\newtheorem{remark}[theorem]{Remark}
\newtheorem{definition}[theorem]{Definition}

\numberwithin{equation}{section}

\newcommand{\N}{\mathbb{N}}

\newcommand{\cG}{\mathcal{G}}

\newcommand{\vC}{\vec{C}}
\newcommand{\vt}{\vec{\theta}}
\newcommand{\vf}{\vec{\phi}}
\newcommand{\Pmic}{P_{\mathrm{mic}}}
\newcommand{\Pcan}{P_{\mathrm{can}}}

\newcommand{\be}{\begin{equation}}
\newcommand{\ee}{\end{equation}}

\title{Covariance structure behind breaking of ensemble equivalence in random graphs}

\author{

Diego Garlaschelli
\footnotemark[1]
\\

Frank den Hollander
\footnotemark[2]
\\

Andrea Roccaverde
\footnotemark[1]\,\,\,\,\footnotemark[2]
}

\footnotetext[1]{
Lorentz Institute for Theoretical Physics, Leiden University, P.O.\  Box 9504, 
2300 RA Leiden, The Netherlands
}

\footnotetext[2]{
Mathematical Institute, Leiden University, P.O.\ Box 9512,
2300 RA Leiden, The Netherlands
}

\date{\today}

\begin{document}

\maketitle 

\begin{abstract}
For a random graph subject to a topological constraint, the \emph{microcanonical ensemble} 
requires the constraint to be met by every realisation of the graph (`hard constraint'), while 
the \emph{canonical ensemble} requires the constraint to be met only on average (`soft 
constraint'). It is known that \emph{breaking of ensemble equivalence} may occur when the 
size of the graph tends to infinity, signalled by a non-zero specific relative entropy of 
the two ensembles. In this paper we analyse a formula for the relative entropy of generic 
discrete random structures recently put forward by Squartini and Garlaschelli. We consider 
the case of a random graph with a given degree sequence (configuration model), and show 
that in the dense regime this formula correctly predicts that the specific relative entropy is 
determined by the scaling of the determinant of the matrix of canonical covariances of the 
constraints. The formula also correctly predicts that an extra correction term is required in 
the sparse regime and in the ultra-dense regime. We further show that the different expressions 
correspond to the degrees in the canonical ensemble being asymptotically \emph{Gaussian} 
in the dense regime and asymptotically \emph{Poisson} in the sparse regime (the latter 
confirms what we found in earlier work), and the dual degrees in the canonical ensemble 
being asymptotically \emph{Poisson} in the ultra-dense regime. In general, we show that 
the degrees follow a multivariate version of the \emph{Poisson-Binomial} distribution in the 
canonical ensemble.



\medskip\noindent
{\it Acknowledgment.} 
DG and AR are supported by EU-project 317532-MULTIPLEX. FdH and AR are supported by NWO 
Gravitation Grant 024.002.003--NETWORKS.
\end{abstract}

\newpage

%%%%%% SECTION 1 %%%%%%%%%%%%%%%%%%

\section{Introduction and main results}
\label{S1}

%%%

\subsection{Background and outline}
\label{S1.1}

For most real-world networks, a detailed knowledge of the architecture of the network
is not available and one must work with a probabilistic description, where the network 
is assumed to be a random sample drawn from a set of allowed configurations that are 
consistent with a set of known \emph{topological constraints}~\cite{SMG15}. Statistical 
physics deals with the definition of the appropriate probability distribution over the set 
of configurations and with the calculation of the resulting properties of the system. Two 
key choices of probability distribution are: 
\begin{itemize}
\item[(1)] 
the \emph{microcanonical ensemble}, where the constraints are \emph{hard} (i.e., are 
satisfied by each individual configuration); 
\item[(2)] 
the \emph{canonical ensemble}, where the constraints are \emph{soft} (i.e., hold as 
ensemble averages, while individual configurations may violate the constraints).
\end{itemize} 
(In both ensembles, the entropy is \emph{maximal} subject to the given constraints.)   

In the limit as the size of the network diverges, the two ensembles are traditionally 
\emph{assumed} to become equivalent, as a result of the expected vanishing of the 
fluctuations of the soft constraints (i.e., the soft constraints are expected to become 
asymptotically hard). However, it is known that this equivalence may be broken, as 
signalled by a non-zero specific relative entropy of the two ensembles (= on an 
appropriate scale). In earlier work various scenarios were identified for this phenomenon 
(see \cite{SdMdHG15}, \cite{GHR17}, \cite{dHMRS17} and references therein). In the 
present paper we take a fresh look at breaking of ensemble equivalence by analysing 
a formula for the relative entropy, based on the \emph{covariance structure} of the 
canonical ensemble, recently put forward by Squartini and Garlaschelli~\cite{GS}. 
We consider the case of a random graph with a given degree sequence (configuration 
model) and show that this formula correctly predicts that the specific relative entropy 
is determined by the scaling of the determinant of the covariance matrix of the constraints 
in the dense regime, while it requires an extra correction term in the sparse regime and 
the ultra-dense regime. We also show that the different behaviours found in the different 
regimes correspond to the degrees being asymptotically Gaussian in the dense regime
and asymptotically Poisson in the sparse regime, and the dual degrees being asymptotically 
Poisson in the ultra-dense regime. We further note that, in general, in the canonical 
ensemble the degrees are distributed according to a multivariate version of the 
\emph{Poisson-Binomial} distribution~\cite{W93}, which admits the Gaussian distribution 
and the Poisson distribution as limits in appropriate regimes.

Our results imply that, in all three regimes, ensemble equivalence breaks down in the presence 
of an extensive number of constraints. This confirms the need for a \emph{principled choice} of 
the ensemble used in practical applications. Three examples serve as an illustration:
\begin{itemize}
\item[(a)]
\emph{Pattern detection} is the identification of nontrivial structural properties in a real-world 
network through comparison with a suitable \emph{null model}, i.e., a random graph model that 
preserves certain local topological properties of the network (like the degree sequence) 
but is otherwise completely random.
\item[(b)] 
\emph{Community detection} is the identification of groups of nodes that are more 
densely connected with each other than expected under a null model, which is a 
popular special case of pattern detection.
\item[(c)]
\emph{Network reconstruction} employs purely local topological information to infer 
higher-order structural properties of a real-world network. This problem arises whenever 
the global properties of the network are not known, for instance, due to confidentiality or 
privacy issues, but local properties are. In such cases, optimal inference about the network 
can be achieved by maximising the entropy subject to the known local constraints, which 
again leads to the two ensembles considered here. 
\end{itemize}
Breaking of ensemble equivalence means that different choices of the ensemble lead to 
asymptotically different behaviours. Consequently, while for applications based on 
ensemble-equi\-valent models the choice of the working ensemble can be arbitrary and 
can be based on mathematical convenience, for those based on ensemble-nonequivalent 
models the choice should be dictated by a criterion indicating which ensemble is the 
appropriate one to use. This criterion must be based on the \emph{a priori} knowledge 
that is available about the network, i.e., which form of the constraint (hard or soft) 
applies in practice.

The remainder of this section is organised as follows. In Section~\ref{S1.2} we define the 
two ensembles and their relative entropy. In Section~\ref{S1.3} we introduce the constraints 
to be considered, which are on the \emph{degree sequence}.  In Section~\ref{S1.4} we 
introduce the various regimes we will be interested in and state a formula for the relative 
entropy when the constraint is on the degree sequence. In Section~\ref{S1.5} we state the 
formula for the relative entropy proposed in \cite{GS} and present our main theorem. In 
Section~\ref{S1.6} we close with a discussion of the interpretation of this theorem and an 
outline of the remainder of the paper. 

%%%

\subsection{Microcanonical ensemble, canonical ensemble, relative entropy}
\label{S1.2}

For $n \in \N$, let $\cG_n$ denote the set of all simple undirected graphs with $n$ nodes.
Any graph $G\in\cG_n$ can be represented as an $n \times n$ matrix with elements 
\be
g_{ij}(G) =
\begin{cases}
1\qquad \mbox{if there is a link between node\ } i \mbox{\ and node\ } j,\\ 
0 \qquad \mbox{otherwise.}
\end{cases}
\ee
Let $\vC$ denote a vector-valued function on $\cG_n$. Given a specific value $\vC^*$, 
which we assume to be \emph{graphical}, i.e., realisable by at least one graph in $\cG_n$, 
the \emph{microcanonical probability distribution} on $\cG_n$ with \emph{hard constraint} 
$\vC^*$ is defined as
\begin{equation}
\Pmic(G) =
\left\{
\begin{array}{ll} 
\Omega_{\vC^*}^{-1}, \quad & \text{if } \vC(G) = \vC^*, \\ 
0, & \text{else},
\end{array}
\right.
\label{eq:PM}
\end{equation}
where 
\begin{equation}
\Omega_{\vC^*} = | \{G \in \cG_n\colon\, \vC(G) = \vC^* \} |
\end{equation}
is the number of graphs that realise $\vC^*$. The \emph{canonical probability distribution} 
$\Pcan(G)$ on $\cG_n$ is defined as the solution of the maximisation of the 
\emph{entropy} 
\begin{equation}
S_n(\Pcan) = - \sum_{G \in \cG_n} \Pcan(G) \ln \Pcan(G)
\end{equation}
subject to the normalisation condition $\sum_{G \in \cG_n} \Pcan(G) = 1$ and to the 
\emph{soft constraint} $\langle \vC \rangle  = \vC^*$, where $\langle \cdot \rangle$ 
denotes the average w.r.t.\ $\Pcan$. This gives
\begin{equation}
\Pcan(G) = \frac{\exp[-H(G,\vt^*)]}{Z(\vt^*)},
\label{eq:PC}
\end{equation}
where 
\begin{equation}
H(G, \vt\,) = \vt \cdot \vC(G)
\label{eq:H}
\end{equation}
is the \emph{Hamiltonian} and
\be
Z(\vt\,) = \sum_{G \in \cG_n} \exp[-H(G, \vt\,)]
\ee
is the \emph{partition function}. In \eqref{eq:PC} the parameter $\vt$ must be set equal to 
the particular value $\vt^*$ that realises $\langle \vC \rangle  = \vC^*$. This value is unique
and maximises the likelihood of the model given the data (see \cite{GL08}).

The \emph{relative entropy} of $\Pmic$ w.r.t.\ $\Pcan$ is~\cite{T15}
\begin{equation}
S_n(\Pmic \mid \Pcan) 
= \sum_{G \in \cG_n} \Pmic(G) \log \frac{\Pmic(G)}{\Pcan(G)},
\label{eq:KL1}
\end{equation}
and the \emph{relative entropy $\alpha_n$-density} is~\cite{GS} 
\begin{equation}
s_{\alpha_n} = {\alpha_n}^{-1}\,S_n(\Pmic \mid \Pcan),
\label{eq:sn}
\end{equation}
where $\alpha_n$ is a \emph{scale parameter}. The limit of the relative entropy 
$\alpha_n$-density is defined as
\begin{equation}
s_{\alpha_\infty}\equiv\lim_{n \to \infty}s_{\alpha_n} 
= \lim_{n \to \infty} {\alpha_n}^{-1}\,S_n(\Pmic \mid \Pcan) \in [0,\infty],
\label{eq:criterion1}
\end{equation}
%%%%%%%%%%%%%%%%%%%%%%%%%%%%%%%%%
%We say that the canonical and microcanonical ensembles are equivalent on scale $\alpha_n$ if and only if 
%\footnote{As shown in \cite{T15} within the context of interacting particle systems, 
%relative entropy is the most sensitive tool to monitor breaking of ensemble equivalence (referred 
%to as breaking \emph{in the measure sense}). Other tools are interesting as well, depending on 
%the `observable' of interest~\cite{T}.}    
%\be
%s_{\alpha_\infty} = 0.
%\label{salphao}
%\ee
%The choice of $\alpha_n$ is flexible. The natural choice is the one for which $s_{\alpha_\infty} 
%\in (0,\infty)$, and depends on the constraint at hand as well as its value. 
%%%%%%%%%%%%%%%%%%%%%%%%%%%%%%%
We say that the microcanonical and canonical ensemble are equivalent \emph{on scale} $\alpha_n$ (or \emph{with speed} 
$\alpha_n$) if and only if \footnote{As shown in \cite{T15} within the context of interacting particle systems, 
relative entropy is the most sensitive tool to monitor breaking of ensemble equivalence (referred 
to as breaking \emph{in the measure sense}). Other tools are interesting as well, depending on 
the `observable' of interest~\cite{T}.}  
\be
s_{\alpha_\infty} = 0.
\label{salphao}
\ee
Clearly, if the ensembles are equivalent with speed $\alpha_n$, then they are also equivalent with any other faster speed $\alpha_n'$ such that $\alpha_n=o(\alpha_n')$. Therefore a natural choice for $\alpha_n$ is the `critical' speed such that the limiting $\alpha_n$-density is positive and finite, i.e. $s_{\alpha_\infty}\in(0,\infty)$. In the following, we will use $\alpha_n$ to denote this natural speed (or scale), and not an arbitrary one. This means that the ensembles are equivalent on all scales faster than $\alpha_n$ and are nonequivalent on scale $\alpha_n$ or slower. The critical scale $\alpha_n$ depends on the constraint at hand as well as its value.
%%%%%%%%%%%%%%%%%%%%%%%%%%%%%%%%%%%%%%%
For instance, if the 
constraint is on the \emph{degree sequence}, then in the sparse regime the natural scale turns 
out to be $\alpha_n=n$ \cite{SdMdHG15}, \cite{GHR17} (in which case $s_{\alpha_\infty}$ is 
the specific relative entropy `per vertex'), while in the dense regime it turns out to be $\alpha_n 
= n\log n$, as shown below. On the other hand, if the constraint is on the \emph{total numbers 
of edges and triangles}, with values different from what is typical for the Erd\H{o}s-Renyi random 
graph in the dense regime, then the natural scale turns out to be $\alpha_n=n^2$ \cite{dHMRS17} 
(in which case $s_{\alpha_\infty}$ is the specific relative entropy `per edge'). Such a severe 
breaking of ensemble equivalence comes from `frustration' in the constraints.  

Before considering specific cases, we recall an important observation made in \cite{SdMdHG15}. 
The definition of $H(G,\vt\,)$ ensures that, for any $G_1,G_2\in\cG_n$, $\Pcan(G_1)=\Pcan(G_2)$ 
whenever $\vC(G_1)=\vC(G_2)$ (i.e., the canonical probability is the same for all graphs having 
the same value of the constraint). We may therefore rewrite \eqref{eq:KL1} as
\begin{equation}
S_n(\Pmic \mid \Pcan) = \log \frac{\Pmic(G^*)}{\Pcan(G^*)},
\label{eq:KL2}
\end{equation}
where $G^*$ is \emph{any} graph in $\cG_n$ such that $\vC(G^*) =\vC^*$ (recall that we 
have assumed that $\vC^*$ is realisable by at least one graph in $\cG_n$). The definition 
in~\eqref{eq:criterion1} then becomes  
\begin{equation}
s_{\alpha_\infty}=\lim_{n \to \infty} {\alpha_n}^{-1}\, \big[\log{\Pmic(G^*)} - \log{\Pcan(G^*)} \big],
\label{eq:criterion3}
\end{equation}
which shows that breaking of ensemble equivalence coincides with $\Pmic(G^*)$ and $\Pcan(G^*)$ 
having different large deviation behaviour on scale $\alpha_n$. Note that \eqref{eq:criterion3} 
involves the microcanonical and canonical probabilities of a \emph{single} configuration $G^*$ 
realising the hard constraint. Apart from its theoretical importance, this fact greatly simplifies 
mathematical calculations. 

To analyse breaking of ensemble equivalence, ideally we would like to be able to identify an 
underlying \emph{large deviation principle} on a natural scale $\alpha_n$. This is generally 
difficult, and so far has only been achieved in the dense regime with the help of \emph{graphons} 
(see \cite{dHMRS17} and references therein). In the present paper we will approach 
the problem from a different angle, namely, by looking at the \emph{covariance matrix of the 
constraints} in the canonical ensemble, as proposed in \cite{GS}.      

Note that all the quantities introduced above in principle depend on $n$. However, except for 
the symbols $\cG_n$ and $S_n(\Pmic \mid \Pcan)$, we suppress the $n$-dependence from 
the notation.

%%%

\subsection{Constraint on the degree sequence}
\label{S1.3}

The degree sequence of a graph $G\in \cG_n$ is defined as $\vec{k}(G) = (k_i(G))_{i=1}^n$ 
with $k_i(G)=\sum_{j \neq i}g_{ij}(G)$. In what follows we constrain the degree sequence to 
a \emph{specific value} $\vec{k}^*$, which we assume to be \emph{graphical}, i.e., there is at 
least one graph with degree sequence $\vec{k}^*$. The constraint is therefore
\be
\label{degcon}
\vC^* = \vec{k}^*= (k_i^*)_{i=1}^n \in \{1,2,\dots ,n-2\}^n,
\ee 
The microcanonical ensemble, when the constraint is on the degree sequence, is known 
as the \emph{configuration model} and has been studied intensively (see 
\cite{SMG15,SdMdHG15,vdH17}). For later use we recall the form of the canonical 
probability in the configuration model, namely,
\begin{equation}
\label{canonical}
\Pcan(G) = \prod_{1 \leq i<j \leq n}\left( p_{ij}^* \right)^{g_{ij}(G)} \left( 1- p_{ij}^* \right)^{1-g_{ij}(G)}
\end{equation}
with 
\be
\label{canonical1} 
p_{ij}^* = \frac{e^{-\theta_i^*-\theta_j^*}}{1 + e^{-\theta_i^*-\theta_j^*}}
\ee 
and with the vector of Lagrange multipliers tuned to the value $\vec{\theta}^*=(\theta_i^*)_{i=1}^n$ 
such that 
\begin{equation}
\label{canonical2}
\langle k_i \rangle = \sum_{j\neq i}p_{ij}^* = k_i^*, \qquad 1\le i\le n.
\end{equation}

Using \eqref{eq:KL2}, we can write
\be
\label{interp}
S_n(\Pmic \mid \Pcan) = \log \frac{\Pmic(G^*)}{\Pcan(G^*)} 
= -\log [\Omega_{\vec{k^*}}\Pcan(G^*)]= -\log Q[\vec{k^*}](\vec{k^*}),
\ee
where $\Omega_{\vec{k}}$ is the number of graphs with degree sequence $\vec{k}$,
\begin{equation}
Q[\vec{k^*}](\vec{k}\,) = \Omega_{\vec{k}}\,\Pcan\big(G^{\vec{k}}\big)
\label{eq:QOmega}
\end{equation}
is the probability that the degree sequence is equal to $\vec{k}$ under the canonical ensemble 
with constraint $\vec{k^*}$, $G^{\vec{k}}$ denotes an arbitrary graph with degree sequence 
$\vec{k}$, and $\Pcan\big(G^{\vec{k}}\big)$ is the canonical probability in~\eqref{canonical} 
rewritten for one such graph:
\be 
\Pcan\big(G^{\vec{k}}\big) 
= \prod_{1 \leq i<j \leq n} \left( p_{ij}^* \right)^{g_{ij}(G^{\vec{k}})} 
\left( 1- p_{ij}^* \right)^{1-g_{ij}(G^{\vec{k}})}
= \prod_{i=1}^n (x_i^*)^{k_i} \prod_{1\le i<j \le n}(1+x_i^* x_j^*)^{-1}.
\label{px}
\ee
In the last expression, $x_i^* = e^{-\theta_i^*}$, and $\vec{\theta}=(\theta_i^*)_{i=1}^n$ is the 
vector of Lagrange multipliers coming from \eqref{canonical1}.

%%%
\subsection{Relevant regimes}
\label{S1.4}

The breaking of ensemble equivalence was analysed in \cite{GHR17} in the so-called 
\emph{sparse regime}, defined by the condition
\be
\max_{1\leq i \leq n} k^*_i = o(\sqrt{n}\,).
\label{eq:sparse1}
\ee
It is natural to consider the opposite setting, namely, the \emph{ultra-dense regime} 
in which the degrees are close to $n-1$, 
\be
\max_{1\leq i \leq n}(n-1-k^*_i) = o(\sqrt{n}\,).
\label{eq:ultradense}
\ee
This can be seen as the \emph{dual} of the sparse regime. We will see in Appendix~\ref{appC} 
that under the map $k^*_i \mapsto n-1-k^*_i$ the microcanonical ensemble and the canonical 
ensemble preserve their relationship, in particular, their relative entropy is invariant.

It is a challenge to study breaking of ensemble equivalence \emph{in between} the sparse 
regime and the ultra-dense regime, called the \emph{dense regime}. In what follows we
 consider a subclass of the dense regime, called the \emph{$\delta$-tame regime}, in which 
the graphs are subject to a certain uniformity condition. 

\begin{definition}
\label{delta}
A degree sequence $\vec{k}^*= (k_i^*)_{i=1}^n$ is called $\delta$-tame if and only if there exists 
a $\delta\in \left(0,\frac{1}{2}\right]$ such that
\begin{equation}
\label{deltatamedef}
\delta \leq p_{ij}^* \leq 1-\delta, \qquad 1\leq i \neq j \leq n,
\end{equation}
where $p_{ij}^*$ are the canonical probabilities in \eqref{canonical}--\eqref{canonical2}.
\end{definition}

\begin{remark}
{\rm The name $\delta$-tame is taken from \cite{BH}, which studies the number of graphs with 
a $\delta$-tame degree sequence. Definition~\ref{delta} is actually a reformulation of the definition 
given in \cite{BH}. See Appendix~\ref{appB} for details.}
\end{remark}

The condition in \eqref{deltatamedef} implies that
\begin{equation}
\label{deltatameondegrees}
(n-1)\delta \leq k_i^* \leq (n-1)(1-\delta), \qquad 1\leq i \leq n,
\end{equation}
i.e., $\delta$-tame graphs are nowhere too thin (sparse regime) nor too dense (ultra-dense regime). 

It is natural to ask whether, conversely, condition \eqref{deltatameondegrees} implies that the degree
sequence is $\delta'$-tame for some $\delta'=\delta'(\delta)$. Unfortunately, this question is not easy
to settle, but the following lemma provides a partial answer. 

\begin{lemma}
\label{reverse}
Suppose that $\vec{k}^*= (k_i^*)_{i=1}^n$ satisfies
\begin{equation}
(n-1)\alpha \leq k_i^* \leq (n-1)(1-\alpha), \qquad 1\leq i \leq n,
\end{equation} 
for some $\alpha \in (\tfrac14,\tfrac12]$. Then there exist $\delta = \delta(\alpha)>0$ and $n_0=n_0(\alpha) 
\in \N$ such that $\vec{k}^*= (k_i^*)_{i=1}^n$ is $\delta$-tame for all $n \geq n_0$.
\end{lemma}

\begin{proof}
The proof follows from \cite[Theorem 2.1]{BH}. In fact, by picking $\beta=1-\alpha$ in that theorem, 
we find that we need $\alpha>\tfrac14$. The theorem also gives information about the values of 
$\delta = \delta(\alpha)$ and $n_0=n_0(\alpha)$. 
\end{proof}

%%%

\subsection{Linking ensemble nonequivalence to the canonical covariances}
\label{S1.5}

In this section we investigate an important formula, recently put forward in~\cite{GS}, for the scaling 
of the relative entropy under a general constraint. The analysis in \cite{GS} allows for the possibility 
that not all the constraints (i.e., not all the components of the vector $\vec{C}$) are linearly independent. 
For instance, $\vec{C}$ may contain redundant replicas of the same constraint(s), or linear combinations 
of them. Since in the present paper we only consider the case where $\vec{C}$ is the degree sequence, 
the different components of $\vec{C}$ (i.e., the different degrees) are linearly independent.

When a $K$-dimensional constraint $\vec{C}^* = (C^*_i)_{i=1}^K$ with independent components is 
imposed, then a key result in~\cite{GS} is the formula
\be
S_n(\Pmic \mid \Pcan) \sim \log\frac{\sqrt{\det(2\pi Q)}}{T}, \qquad n\to\infty,
\ee
where
\begin{equation}
\label{covQ}
Q=(q_{ij})_{1 \leq i,j \leq K}
\end{equation}
is the $K\times K$ covariance matrix of the constraints under the canonical ensemble, whose entries are 
defined as
\be
\qquad q_{ij} = \mathrm{Cov}_{\Pcan}(C_i,C_j)=\langle C_i\,C_j\rangle-\langle C_i\rangle \langle C_j\rangle,
\ee
and
\be
T=\prod_{i=1}^K\left[1+O\left(1/\lambda_i^{(K)}(Q)\right)\right], 
\label{eq:T}
\ee
with $\lambda_i^{(K)}(Q)>0$ the $i$-th eigenvalue of the $K\times K$ covariance matrix $Q$. This result can 
be formulated rigorously as
\begin{formula}[\cite{GS}]
\label{conj}
If all the constraints are linearly independent, then the limiting relative entropy ${\alpha_n}$-density 
equals
\be
s_{\alpha_\infty}=\lim_{n\to\infty}\frac{\log\sqrt{\det(2\pi Q)}}{\alpha_n}+\tau_{\alpha_\infty}
\label{eq:deltalimit}
\ee
with $\alpha_n$ the `natural' speed and
\be
\tau_{\alpha_\infty}=-\lim_{n\to\infty}\frac{\log T}{\alpha_n}.
\label{eq:tau}
\ee
The latter is zero when 
\begin{equation}
\label{CondforFormula}
\lim_{n\to\infty} \frac{|I_{K_n,R}|}{\alpha_n}=0\quad \forall\,R<\infty,
\end{equation}
where $I_{K,R} = \lbrace i=1,\dots,K\colon\,\lambda_i^{(K)}(Q) \le R \rbrace$ with $\lambda_i^{(K)}(Q)$
the $i$-th eigenvalue of the $K$-dimensional covariance matrix $Q$ (the notation $K_n$ indicates that
$K$ may depend on $n$). Note that $0\le I_{K,R} \le K$. Consequently, \eqref{CondforFormula} is satisfied 
(and hence $\tau_{\alpha_\infty}=0$) when $\lim_{n\to\infty} K_n/\alpha_n=0$, i.e., when the number $K_n$ 
of constraints grows slower than $\alpha_n$.
\end{formula}

\begin{remark}[\cite{GS}]
{\rm Formula~{\rm \ref{conj}}, for which \cite{GS} offers compelling evidence but not a mathematical 
proof, can be rephrased by saying that the natural choice of $\alpha_n$
is
\be
\tilde{\alpha}_n=\log\sqrt{\det(2\pi Q)}.
\label{eq:alphatilde}
\ee
Indeed, if all the constraints are linearly independent and \eqref{CondforFormula} holds, 
then $\tau_{\tilde{\alpha}_n}=0$ and 
\begin{eqnarray}
\label{eq:deltalimit2}
&s_{\tilde{\alpha}_\infty}=1,\\
&S_n(\Pmic \mid \Pcan)=[1+o(1)]\,\tilde{\alpha}_n.
\label{conjecture2}
\end{eqnarray}
}
\end{remark}

We now present our main theorem, which considers the case where the constraint is on the 
degree sequence: $K_n=n$ and $\vC^*=\vec{k}^*= (k_i^*)_{i=1}^n$. This case was studied in \cite{GHR17}, 
for which $\alpha_n = n$ in the \emph{sparse regime with finite degrees}. Our results here 
focus on three new regimes, for which we need to increase $\alpha_n$: the \emph{sparse 
regime with growing degrees}, the \emph{$\delta$-tame regime}, and the \emph{ultra-dense 
regime with growing dual degrees}.
In all these cases, since $\lim_{n\to\infty} K_n/\alpha_n=\lim_{n\to\infty} n/\alpha_n=0$, {\rm Formula~{\rm \ref{conj}} states that~\eqref{eq:deltalimit} holds with $\tau_{\tilde{\alpha}_n}=0$.
Our theorem provides a rigorous and independent mathematical proof of this result.

\begin{theorem}
\label{MainTheorem}
Formula~{\rm \ref{conj}} is true with $\tau_{\alpha_\infty}=0$ when the constraint is on the 
degree sequence $\vC^*= \vec{k}^*= (k_i^*)_{i=1}^n$, the scale parameter is $\alpha_n 
= n\,\overline{f_n}$ with
\be
\label{barfndef}
\overline{f_n} = n^{-1} \sum_{i=1}^n f_n(k_i^*) \quad \text{ with } \quad 
f_n(k)=\frac{1}{2}\log\left[\frac{k(n-1-k)}{n}\right],
\ee 
and the degree sequence belongs to one of the following three regimes: 
\begin{itemize}
\item
The sparse regime with growing degrees:
\begin{equation}
\label{kmin}
\max_{1\leq i \leq n} k^*_i = o(\sqrt{n}\,),\qquad
\lim_{n\to\infty}\min_{1\leq i \leq n} k^*_i = \infty.
\end{equation}
\item
The $\delta$-tame regime (see \eqref{canonical} and Lemma~{\rm \ref{reverse}}):
\begin{equation}
\delta \leq p_{ij}^* \leq 1-\delta, \quad 1 \leq i\neq j \leq n.
\end{equation}
\item
The ultra-dense regime with growing dual degrees:
\begin{equation}
\max_{1\leq i \leq n}(n-1 - k^*_i) = o(\sqrt{n}\,),\qquad
\lim_{n\to\infty} \min_{1\leq i \leq n} (n-1-k^*_i) = \infty.
\end{equation}
\end{itemize}
In all three regimes there is breaking of ensemble equivalence, and  
\begin{equation}
\label{Result}
s_{{\alpha}_\infty}= \lim_{n \to \infty} s_{\alpha_n} = 1.
\end{equation}
\end{theorem}

%%%%

\subsection{Discussion and outline}
\label{S1.6}

Comparing \eqref{eq:deltalimit2} and~\eqref{Result}, and using~\eqref{eq:alphatilde}, 
we see that Theorem~\ref{MainTheorem} shows that if the constraint is on the degree 
sequence, then
\be
\label{rescaling}
S_n(\Pmic \mid \Pcan) \sim n \overline{f_n}\sim\log\sqrt{\det(2\pi Q)}
\ee 
in each of the three regimes considered. Below we provide a heuristic explanation for this 
result (as well as for our previous results in~\cite{GHR17}) that links back to \eqref{interp}. 
In Section \ref{S2} we prove Theorem \ref{MainTheorem}.

\paragraph{Poisson-Binomial degrees in the general case.}
Note that~\eqref{interp} can be rewritten as
\begin{equation}
S_n(\Pmic \mid \Pcan) = S\big(\,\delta[\vec{k^*}] \mid Q[\vec{k^*}]\,\big),
\label{SQ}
\end{equation} 
where $\delta[\vec{k^*}] = \prod_{i=1}^n \delta[k^*_i]$ is the \emph{multivariate Dirac distribution} 
with average $\vec{k^*}$. This has the interesting interpretation that the relative entropy between 
the distributions $\Pmic$ and $\Pcan$ \emph{on the set of graphs} coincides with the relative 
entropy between $\delta[\vec{k^*}]$ and $Q[\vec{k^*}]$ \emph{on the set of degree sequences}.

To be explicit, using \eqref{eq:QOmega} and~\eqref{px}, we can rewrite $Q[\vec{k^*}](\vec{k})$ 
as 
\begin{equation}
Q[\vec{k^*}](\vec{k}) =\Omega_{\vec{k}}\ \prod_{i=1}^n (x_i^*)^{k_i}
\prod_{1\le i<j \le n}(1+x_i^* x_j^*)^{-1}.
\label{eq:PoissonBinomial}
\end{equation}
We note that the above distribution is a multivariate version of the \emph{Poisson-Binomial 
distribution} (or Poisson's Binomial distribution; see Wang~\cite{W93}). In the univariate case, 
the Poisson-Binomial distribution describes the probability of a certain number of successes 
out of a total number of independent and (in general) not identical Bernoulli trials~\cite{W93}. 
In our case, the marginal probability that node $i$ has degree $k_i$ in the canonical ensemble, 
irrespectively of the degree of any other node, is indeed a univariate Poisson-Binomial given 
by $n-1$ independent Bernoulli trials with success probabilities $\{p_{ij}^*\}_{j\ne i}$. 
The relation in ~\eqref{SQ} can therefore be restated as
\begin{equation}
S_n(\Pmic \mid \Pcan) = S\big(\,\delta[\vec{k^*}] 
\mid \mathrm{PoissonBinomial}[\vec{k^*}]\,\big), 
\end{equation}
where $\mathrm{PoissonBinomial}[\vec{k^*}]$ is the multivariate Poisson-Binomial distribution 
given by ~\eqref{eq:PoissonBinomial}, i.e.,
\be
Q[\vec{k^*}] = \mathrm{PoissonBinomial}[\vec{k^*}].
\ee
The relative entropy can therefore be seen as coming from a situation in which the microcanonical 
ensemble forces the degree sequence to be exactly $\vec{k^*}$, while the canonical ensemble 
forces the degree sequence to be Poisson-Binomial distributed with average $\vec{k^*}$.

It is known that the univariate Poisson-Binomial distribution admits two asymptotic limits: (1) a 
Poisson limit (if and only if, in our notation, $\sum_{j\ne i}p_{ij}^*\to\lambda>0$ and $\sum_{j\ne i}
(p_{ij}^*)^2\to0$ as $n\to\infty$~\cite{W93}); (2) a Gaussian limit (if and only if $p_{ij}^*\to\lambda_j>0$ 
for all $j\ne i$ as $n\to\infty$, as follows from a central limit theorem type of argument). If all 
the Bernoulli trials are identical, i.e., if all the probabilities $\{p_{ij}^*\}_{j\ne i}$ are equal, then 
the univariate Poisson-Binomial distribution reduces to the ordinary Binomial distribution, which 
also exhibits the well-known Poisson and Gaussian limits. These results imply that also the 
general multivariate Poisson-Binomial distribution in \eqref{eq:PoissonBinomial} admits limiting 
behaviours that should be consistent with the Poisson and Gaussian limits discussed above for
its marginals. This is precisely what we confirm below.

\paragraph{Poisson degrees in the sparse regime.}
In \cite{GHR17} it was shown that, for a sparse degree sequence,
\be
\label{interpretation1}
S_n(\Pmic \mid \Pcan) \sim \sum_{i=1}^n S\big(\,\delta[k^*_i] 
\mid \mathrm{Poisson}[k^*_i]\,\big). 
\ee
The right-hand side is the sum over all nodes $i$ of the relative entropy of the \emph{Dirac 
distribution} with average $k^*_i$ w.r.t.\ the \emph{Poisson distribution} with average $k^*_i$. 
We see that, under the sparseness condition, the constraints act on the nodes essentially 
independently. 
We can therefore reinterpret \eqref{interpretation1} as the statement
\begin{equation}
\label{interpretation2}
S_n(\Pmic \mid \Pcan) \sim S\big(\,\delta[\vec{k^*}] 
\mid \mathrm{Poisson}[\vec{k^*}]\,\big), 
\end{equation}
where $\mathrm{Poisson}[\vec{k^*}]$ $= \prod_{i=1}^n \mathrm{Poisson}[k^*_i]$ is the 
\emph{multivariate Poisson distribution} with average $\vec{k^*}$. In other words, in this regime
\be
\label{Psp}
Q[\vec{k^*}] \sim \mathrm{Poisson}[\vec{k^*}],
\ee
i.e. the joint multivariate Poisson-Binomial distribution~\eqref{eq:PoissonBinomial} essentially 
decouples into the product of marginal univariate Poisson-Binomial distributions describing 
the degrees of all nodes, and each of these Poisson-Binomial distributions is asymptotically 
a Poisson distribution.

Note that the Poisson regime was obtained in~\cite{GHR17} under the condition in~\eqref{eq:sparse1}, 
which is less restrictive than the aforementioned condition $k_i^*=\sum_{j\ne i}p_{ij}^*\to\lambda>0$, 
$\sum_{j\ne i}(p_{ij}^*)^2\to0$ under which the Poisson distribution is retrieved from the Poisson-Binomial 
distribution~\cite{W93}. In particular, the condition in~\eqref{eq:sparse1} includes both the case 
with growing degrees included in Theorem~\ref{MainTheorem} (and consistent with Formula~{\rm \ref{conj}} 
with $\tau_{\alpha_\infty}=0$) and the case with finite degrees, which \emph{cannot} be retrieved from 
Formula~{\rm \ref{conj}} with $\tau_{\alpha_\infty}=0$, because it corresponds to the case where all the $n=\alpha_n$ eigenvalues of 
$Q$ remain finite as $n$ diverges (as the entries of $Q$ themselves do not diverge), and indeed~\eqref{CondforFormula} does not hold.

\paragraph{Poisson degrees in the ultra-dense regime.}
Since the ultra-dense regime is the dual of the sparse regime, we immediately get the heuristic 
interpretation of the relative entropy when the constraint is on an ultra-dense degree sequence 
$\vec{k}^*$. Using \eqref{interpretation2} and the observations in Appendix \ref{appC} (see, in
particular \eqref{Ultrarelativeentropy}), we get
\begin{equation}
\label{interpretation22}
S_n(\Pmic \mid \Pcan) \sim S\big(\,\delta[\vec{\ell^*}] \mid \mathrm{Poisson}[\vec{\ell^*}]\,\big), 
\end{equation}
where $\vec{\ell}^*= (\ell_i^*)_{i=1}^n$ is the dual degree sequence given by $\ell_i^* = n-1-k_i^*$.
In other words, under the microcanonical ensemble the dual degrees follow the distribution 
$\delta[\vec{\ell^*}]$, while under the canonical ensemble the dual degrees follow the distribution 
$Q[\vec{\ell^*}]$, where in analogy with \eqref{Psp},
\be
\label{Pud}
Q[\vec{\ell^*}] \sim \mathrm{Poisson}[\vec{\ell^*}].
\ee
Similar to the sparse case, the multivariate Poisson-Binomial distribution~\eqref{eq:PoissonBinomial} 
reduces to a product of marginal, and asymptotically Poisson, distributions governing the different 
degrees.

Again, the case with finite dual degrees cannot be retrieved from Formula~{\rm \ref{conj}} with 
$\tau_{\alpha_\infty}=0$, because it corresponds to the case where $Q$ has a diverging 
(like $n=\alpha_n$) number of eigenvalues whose value remains finite as $n\to\infty$, and~\eqref{CondforFormula} does not hold. By contrast, the case with growing 
dual degrees can be retrieved from Formula~{\rm \ref{conj}} with $\tau_{\alpha_\infty}=0$ because~\eqref{CondforFormula} holds, as 
confirmed in Theorem~\ref{MainTheorem}.

\paragraph{Gaussian degrees in the dense regime.}
We can reinterpet \eqref{rescaling} as the statement
\be
\label{extra} 
S_n(\Pmic \mid \Pcan) \sim S\big(\,\delta[\vec{k^*}] \mid \mathrm{Normal}[\vec{k^*},Q]\,\big),
\ee
where $\mathrm{Normal}[\vec{k^*},Q]$ is the \emph{multivariate Normal distribution} with 
mean $\vec{k^*}$ and covariance matrix $Q$. In other words, in this regime
\be 
Q[\vec{k^*}] \sim \mathrm{Normal}[\vec{k^*},Q],
\ee
i.e., the multivariate Poisson-Binomial distribution~\eqref{eq:PoissonBinomial} is asymptotically 
a multivariate Gaussian distribution whose covariance matrix is in general not diagonal, i.e., the 
dependencies between degrees of different nodes do \emph{not} vanish, unlike in the other two 
regimes. Since all the degrees are growing in this regime, so are all the eigenvalues of $Q$, implying~\eqref{CondforFormula} and
consistently with {\rm Formula~{\rm \ref{conj}} with $\tau_{\alpha_\infty}=0$, as proven 
in Theorem~\ref{MainTheorem}. 

Note that the right-hand side of \eqref{extra}, being the relative entropy of a discrete distribution 
with respect to a continuous distribution, needs to be properly interpreted: the Dirac distribution
$\delta[\vec{k^*}]$ needs to be smoothened to a continuous distribution with support in a small
ball around $\vec{k^*}$. Since the degrees are large, this does not affect the asymptotics.

\paragraph{Crossover between the regimes.}
An easy computation gives 
\be
\label{gkdef}
S\big(\,\delta[k^*_i] \mid \mathrm{Poisson}[k^*_i]\,\big) = g(k^*_i)
\quad \text{ with } \quad g(k) = \log\left(\frac{k!}{e^{-k}k^k}\right).
\ee
Since $g(k) = [1+o(1)]\tfrac12 \log (2\pi k)$, $k\to\infty$, we see that, as we move from the sparse 
regime with finite degrees to the sparse regime with growing degrees, the scaling of the relative 
entropy in~\eqref{interpretation1} nicely links up with that of the dense regime in~\eqref{extra} 
via the common expression in~\eqref{rescaling}. Note, however, that since the sparse regime 
with growing degrees is in general incompatible with the dense $\delta$-tame regime, in 
Theorem~\ref{MainTheorem} we have to obtain the two scalings of the relative entropy under 
disjoint assumptions. By contrast, Formula~{\rm \ref{conj}} with $\tau_{\alpha_\infty}=0$, and 
hence \eqref{conjecture2}, unifies the two cases under the simpler and more general requirement 
that all the eigenvalues of $Q$, and hence all the degrees, diverge. Actually, \eqref{conjecture2} 
is expected to hold in the even more general hybrid case where there are both finite and growing degrees, 
provided the number of finite-valued eigenvalues of $Q$ grows slower than $\alpha_n$~\cite{GS}.

\paragraph{Other constraints.}
It would be interesting to investigate Formula~\ref{conj} for constraints other than on the 
degrees. Such constraints are typically much harder to analyse. In \cite{dHMRS17} constraints
are considered on the total number of edges and the total number of triangles \emph{simultaneously} ($K=2$)
in the dense regime. It was found that, with $\alpha_n=n^2$, breaking of ensemble equivalence 
occurs for some `frustrated' choices of these numbers. 
Clearly, this type of breaking of ensemble equivalence does not arise from the recently proposed~\cite{GS} mechanism associated with a diverging number of constraints as in the cases considered in this paper, but from the more traditional~\cite{T15} mechanism of a phase transition associated with the frustration phenomenon.

\paragraph{Outline.}
Theorem~\ref{MainTheorem} is proved in Section~\ref{S2}. In Appendix~\ref{appB} we show that 
the canonical probabilities in \eqref{canonical} are the same as the probabilities used in \cite{BH} 
to define a $\delta$-tame degree sequence. In Appendix~\ref{appC} we explain the duality 
between the sparse regime and the ultra-dense regime.

%%%%%%%% SECTION 2 %%%%%%%%%%%%%%%%%%%%%%%%

\section{Proof of the main theorem}
\label{S2}

In Section~\ref{S2.2} we prove Theorem~\ref{MainTheorem}. The proof is based on two lemmas,
which we state and prove in Section~\ref{S2.1}.

%%%

\subsection{Preparatory lemmas}
\label{S2.1}

The following lemma gives an expression for the relative entropy. 

\begin{lemma}
\label{lemma2}
If the constraint is a $\delta$-tame degree sequence, then the relative entropy in \eqref{eq:KL2} 
scales as
\begin{equation}
\label{RelEntro}
S_n(\Pmic \mid \Pcan) = [1+o(1)]\,\tfrac12\log[\det(2\pi Q)],
\end{equation}
where $Q$ is the covariance matrix in \eqref{covQ}. This matrix $Q=(q_{ij})$ takes the form 
\begin{equation}
\label{Q}
\begin{cases}
q_{ii} = k_i^*-\sum_{j \neq i}(p_{ij}^*)^2 
= \sum_{j \neq i} p_{ij}^*(1-p_{ij}^*), \quad 1 \leq i \leq n,\\
q_{ij} = p_{ij}^*(1-p_{ij}^*), \quad 1 \leq i \neq j \leq n.
\end{cases}
\end{equation}
\end{lemma}

\begin{proof}
To compute $q_{ij}=\mathrm{Cov}_{\Pcan}(k_i,k_j)$ we take the second order derivatives of 
the log-likelihood function
\be
{\cal{L}}(\vec{\theta}) = \log\Pcan(G^* \mid \vec{\theta})
= \log\left[ \prod_{1 \leq i < j \leq n} p_{ij}^{g_{ij}(G^*)} (1-p_{ij})^{(1-g_{ij}(G^*))} \right],
\quad p_{ij} = \frac{e^{-\theta_i - \theta_j}}{1+e^{-\theta_i - \theta_j}}
\ee
in the point $\vec{\theta}=\vec{\theta}^*$~\cite{GS}. 
Indeed, it is easy to show that the first-order derivatives are~\cite{GL08}
\be
\frac{\partial}{\partial \theta_i}{\cal{L}}(\vec{\theta}\,)
= \langle k_i\rangle - k_i^*,
\quad 
\frac{\partial}{\partial \theta_i}{\cal{L}}(\vec{\theta}\,)\bigg|_{\vec{\theta}=\vec{\theta^*}}
= k_i^*-k_i^*=0
\ee
and the second-order derivatives are
\begin{equation}
\label{Covcalc}
\frac{\partial^2}{\partial \theta_i\partial \theta_j}{\cal{L}}(\vec{\theta})
\bigg|_{\vec{\theta}=\vec{\theta^*}} =\langle k_i\,k_j\rangle - \langle k_i\rangle\langle k_j\rangle
= \mathrm{Cov}_{\Pcan}(k_i,k_j).
\end{equation}
This readily gives~\eqref{Q}.

%%%%
The proof of \eqref{RelEntro} uses \cite[Eq.~(1.4.1)]{BH}, which says that if a $\delta$-tame degree 
sequence is used as constraint, then
\begin{equation}
\label{relentlemma}
\Pmic^{-1}(G^*) = \Omega_{\vec{C}^*} = \frac{e^{H(p^*)}}{(2\pi)^{n/2}\sqrt{\det(Q)}}\ e^{C},
\end{equation}
where $Q$ and $p^*$ are defined in \eqref{Q} and \eqref{pbm} below, while $e^C$ is sandwiched 
between two constants that depend on $\delta$:
\begin{equation}
\gamma_1(\delta) \leq e^C \le \gamma_2(\delta).
\end{equation} 
From \eqref{relentlemma} and the relation $H(p^*) = -\log\Pcan(G^*)$, proved in Lemma~\ref{lemma1}
below, we get the claim.
\end{proof}

The following lemma shows that the diagonal approximation of $\log(\det Q)/n\overline{f}_n$ is good 
when the degree sequence is $\delta$-tame.

\begin{lemma}
\label{lemma3}
Under the $\delta$-tame condition,
\begin{equation}
\label{QQDcomp}
\log(\det Q_D) + o(n\,\overline{f}_n) \leq \log(\det Q) \leq \log(\det Q_D)
\end{equation}
with $Q_D=\mathrm{diag}(Q)$ the matrix that coincides with $Q$ on the diagonal 
and is zero off the diagonal.
\end{lemma}

\begin{proof}
We use \cite[Theorem 2.3]{IL03}, which says that if
\begin{itemize}
\item[(1)] $\det(Q)$ is real, 
\item[(2)] $Q_D$ is non-singular with $\det(Q_D)$ real,
\item[(3)] $\lambda_i (A)>-1$, $1 \leq i \leq n$,
\end{itemize}
then 
\begin{equation}
\label{ILbds}
e^{-\frac{n\rho^2(A)}{1+\lambda_{\min}(A)}} \det Q_D \leq \det Q \leq \det Q_D.
\end{equation}
Here, $A=Q_D^{-1}Q_{\mathrm{off}}$, with $Q_{\mathrm{off}}$ the matrix that coincides with 
$Q$ off the diagonal and is zero on the diagonal, $\lambda_i(A)$ is the $i$-th eigenvalue of $A$
(arranged in decreasing order), $\lambda_{\mathrm{min}}(A) = \min _{1 \leq i \leq n}\lambda_i(A)$, 
and $\rho(A) = \max_{1 \leq i \leq n}|\lambda_i(A)|$.

We begin by verifying (1)--(3). 

\medskip\noindent
(1) Since $Q$ is a symmetric matrix with real entries, $\det Q$ exists and is real.

\medskip\noindent
(2) This property holds thanks to the $\delta$-tame condition. Indeed, since $q_{ij} = p_{i,j}^*(1-p_{i,j}^*)$, 
we have 
\begin{equation}
\label{qij}
0 < \delta^{2} \leq q_{ij} \leq (1-\delta)^{2} < 1,
\end{equation}
which implies that
\begin{equation}
\label{qii}
0 < (n-1)\delta^2 \leq q_{ii} = \sum_{j\neq i} q_{ij} \leq (n-1)(1-\delta)^2.  
\end{equation}

\medskip\noindent
(3) It is easy to show that $A=(a_{ij})$ is given by
\begin{equation}
\label{aijrel}
a_{ij} = \left\{\begin{array}{ll}
\frac{q_{ij}}{q_{ii}}, &1 \leq i \neq j \leq n,\\
0, &1 \leq i \leq n,
\end{array}
\right.
\end{equation}
where $q_{ij}$ is given by \eqref{Q}. Since $q_{ij}=q_{ji}$, the matrix $A$ is symmetric. Moreover, 
since $q_{ii} = \sum_{j\neq i} q_{ij}$, the matrix $A$ is also Markov. We therefore have 
\begin{equation}
\label{Markov}
1 = \lambda_1(A) \geq \lambda_2(A) \geq \dots \geq \lambda_n(A) \geq -1.
\end{equation}
From \eqref{qij} and \eqref{aijrel} we get 
\begin{equation}
\label{aij}
0 < \frac{1}{n-1} \left(\frac{\delta}{1-\delta}\right)^2 \leq a_{ij} 
\leq \frac{1}{n-1}\left(\frac{1-\delta}{\delta}\right)^2.
\end{equation}
This implies that the Markov chain on $\left\lbrace 1,\dots,n\right\rbrace$ with transition 
matrix $A$ starting from $i$ can return to $i$ with a positive probability after an arbitrary 
number of steps $\geq 2$. Consequently, the last inequality in \eqref{Markov} is strict.

\medskip
We next show that 
\be
\frac{n\rho^2(A)}{1+\lambda_{\min}(A)} = o(n\,\overline{f}_n).
\ee 
Together with \eqref{ILbds} this will settle the claim in \eqref{QQDcomp}. From 
\eqref{Markov} it follows $\rho(A) = 1$, so we must show that
\begin{equation}
\label{tbs}
\lim_{n\to \infty} [1+\lambda_{\min}(A)]\,\overline{f}_n = \infty.
\end{equation}
Using \cite[Theorem 4.3]{Z04}, we get 
\begin{equation}
\label{Zhangeq}
\lambda_{\min}(A) \geq -1 + \frac{\min_{1 \leq i \neq j \leq n } \pi_ia_{ij}}
{\min_{1 \leq i \leq n} \pi_i}\,\mu_{\mathrm{min}}(L) + 2\gamma.
\end{equation}
Here, $\pi=(\pi_i)_{i=1}^n$ is the invariant distribution of the reversible Markov chain
with transition matrix $A$, while $\mu_{\min}(L)=\min_{1 \leq i \leq n} \lambda_i (L)$ 
and $\gamma = \min_{1 \leq i \leq n} a_{ii}$, with $L = (L_{ij})$ the matrix such 
that, for $i \neq j$, $L_{ij}=1$ if and only if $a_{ij} > 0$, while $L_{ii} = \sum_{j\neq i} 
L_{ij}$. 

We find that $\pi_i = \frac{1}{n}$ for $1 \leq i \leq n$, $L_{ij}=1$ for $1 \leq i \neq j \leq n$, 
$L_{ii} = n-1$ for $1 \leq i \leq n$, and $\gamma = 0$. Hence \eqref{Zhangeq} becomes
\begin{equation}
\lambda_{\mathrm{min}}(A) \geq -1 + (n-2) \min_{1 \leq i \neq j \leq n} a_{ij} 
\geq -1 + \frac{n-2}{n-1}\left(\frac{\delta}{1-\delta}\right)^2,
\end{equation}
where the last inequality comes from \eqref{aij}. To get \eqref{tbs} it therefore suffices to show 
that $\overline{f}_{\infty} = \lim_{n\to\infty}\overline{f}_n=\infty$. But, using the $\delta$-tame 
condition, we can estimate
\be
\label{ineqlogki}
\begin{aligned}
&\frac{1}{2}\log\left[\frac{(n-1)\delta(1-\delta+n\delta)}{n}\right] 
\leq \overline{f}_n = \frac{1}{2n} \sum_{i=1}^n \log\left[\frac{k_i^*(n-1-k_i^*)}{n}\right]\\ 
&\qquad\leq \frac{1}{2}\log\left[\frac{(n-1)(1-\delta)(\delta + n(1-\delta))}{n}\right],
\end{aligned}
\ee
and both bounds scale like $\frac{1}{2}\log n$ as $n\to\infty$.
\end{proof}

%%%

\subsection{Proof of Theorem~\ref{MainTheorem}}
\label{S2.2}

\begin{proof}
We deal with each of the three regimes in Theorem~\ref{MainTheorem} separatetely.

\paragraph{The sparse regime with growing degrees.}
 
Since $\vec{k}^*= (k_i^*)_{i=1}^n$ is a  sparse degree sequence, we can use \cite[Eq.~(3.12)]{GHR17},
which says that
\begin{equation}
\label{sparserelent}
S_n(\Pmic \mid \Pcan) = \sum_{i=1}^n g(k_i^*) + o(n), \qquad n\to\infty,
\end{equation}
where $g(k)=\log \left(\frac{k!}{k^k e^{-k}}\right)$ is defined in \eqref{gkdef}. Since the degrees are 
growing, we can use Stirling's approximation $g(k) = \tfrac12\log(2\pi k) + o(1)$, $k\to\infty$, to obtain
\begin{equation}
\label{gki}
\sum_{i=1}^n g(k_i^*) = \tfrac{1}{2}\sum_{i=1}^n\log \left( 2\pi k_i^* \right) + o(n) 
= \tfrac{1}{2} \left[n \log 2\pi + \sum_{i=1}^n \log k_i^*\right] + o(n).
\end{equation}
Combining \eqref{sparserelent}--\eqref{gki}, we get
\begin{equation}
\label{Slim1}
\frac{S_n(\Pmic \mid \Pcan)}{n\,\overline{f}_n} 
= \tfrac{1}{2} \left[ \frac{\log 2\pi}{\overline{f}_n} 
+ \frac{\sum_{i=1}^n \log k_i^*}{n\overline{f}_n} \right] + o(1).
\end{equation}
Recall \eqref{barfndef}. Because the degrees are sparse, we have
\be
\label{Slim2}
\lim_{n\to\infty} \frac{\sum_{i=1}^n \log k_i^*}{n\overline{f}_n} = 2.
\ee
Because the degrees are growing, we also have
\be
\label{Slim3}
\overline{f}_{\infty} = \lim_{n\to\infty}\overline{f}_n =\infty.
\ee
Combining \eqref{Slim1}--\eqref{Slim3} we find that $\lim_{n\to\infty} S_n(\Pmic \mid \Pcan)
/n\,\overline{f}_n = 1$.

\paragraph{The ultra-dense regime with growing dual degrees.} 

If $\vec{k}^*= (k_i^*)_{i=1}^n$ is an ultra-dense degree sequence, then the dual $\vec{\ell}^* 
= (\ell_i^*)_{i=1}^n = (n-1-k_i^*)_{i=1}^n$ is a sparse degree sequence. By Lemma 
\ref{Ultrarelativeentropy}, the relative entropy is invariant under the map $k_i^* \to \ell_i^*
= n-1-k_i^*$. So is $\bar{f_n}$, and hence the claim follows from the proof in the sparse
regime.

\paragraph{The $\delta$-tame regime.}

It follows from Lemma~\ref{lemma2} that
\begin{equation}
\label{finalrelent}
\lim_{n \to \infty} \frac{S_n(\Pmic \mid \Pcan)}{n\,\overline{f}_n} 
= \tfrac{1}{2}\left[\lim_{n \to \infty}\frac{\log 2\pi}{\overline{f}_n} 
+ \lim_{n \to \infty}\frac{\log(\det Q)}{n\,\overline{f}_n}\right].
\end{equation}
From \eqref{ineqlogki} we know that $\overline{f}_{\infty} = \lim_{n\to\infty}\overline{f}_n=\infty$ 
in the \emph{$\delta$-tame regime}. It follows from Lemma~\ref{lemma3} that
\begin{equation}
\lim_{n \to \infty} \frac{\log(\det Q)}{n\,\overline{f}_n} 
= \lim_{n \to \infty}\frac{\log(\det Q_D)}{n\,\overline{f}_n}.
\end{equation}
To conclude the proof it therefore suffices to show that
\begin{equation}
\label{qdnf}
\lim_{n \to \infty} \frac{\log(\det Q_D)}{n\,\overline{f}_n} = 2.
\end{equation} 
Using \eqref{qii} and \eqref{ineqlogki}, we may estimate
\begin{equation}
\label{qdnf2}
\frac{2\log[(n-1)\delta^2]}{\log \frac{(n-1)(1-\delta)(\delta + n(1-\delta))}{n}} 
\leq \frac{\sum_{i=1}^n\log(q_{ii})}{n\,\overline{f}_n} 
= \frac{\log(\det Q_D)}{n\,\overline{f}_n} 
\leq \frac{2\log[(n-1)(1-\delta)^2]}{\log \frac{(n-1)\delta(1-\delta+n\delta)}{n}}.
\end{equation}
Both sides tend to 2 as $n\to\infty$, and so \eqref{qdnf} follows. 
\end{proof}

%%%%%%%%%%%%%%%%%%%%%%%%%%%%%%%%%%%%%%%%%%%

\appendix

%%%%%%%%%%%%% APPENDIX A %%%%%%%%%%%%%%%%%%

\section{Appendix}
\label{appB}

Here we show that the canonical probabilities in \eqref{canonical} are the same as the
probabilities used in \cite{BH} to define a $\delta$-tame degree sequence. 

For $q = (q_{ij})_{1 \leq i,j\leq n}$, let 
\begin{equation}
E(q) = -\sum_{1 \leq i \neq j \leq n} q_{ij}\log q_{ij}+ (1-q_{ij})\log(1-q_{ij}).
\end{equation}
be the entropy of $q$. For a given degree sequence $(k_i^*)_{i=1}^n$, consider the 
following maximisation problem:
\begin{equation}
\label{pbm}
\begin{cases}
\max E(q),\\
\sum_{j\neq i} q_{ij}=k_i^*,\,\,1 \leq i \leq n,\\
0 \leq q_{ij} \leq 1,\,\, 1\leq i \neq j \leq n.
\end{cases}
\end{equation}
Since $q \mapsto E(q)$ is strictly concave, it attains its maximum at a unique point. 

\begin{lemma}
\label{lemma1}
The canonical probability takes the form 
\begin{equation}
\label{Pcan}
\Pcan(G) = \prod_{1 \leq i < j \leq n} \left(p_{ij}^*\right)^{g_{ij}(G)}\left(1-p_{ij}^*\right)^{1-g_{ij}(G)},
\end{equation}
where $p^*=(p_{ij}^*)$ solves \eqref{pbm}. In addition,
\begin{equation}
\label{logPcan}
\log\Pcan(G^*) = -H(p^*).
\end{equation}
\end{lemma}

\begin{proof}
It was shown in \cite{GHR17} that, for a degree sequence constraint,
\begin{equation}
\Pcan(G) = \prod_{1 \leq i < j \leq n}\left(p_{ij}^*\right)^{g_{ij}(G)}\left(1-p_{ij}^*\right)^{1-g_{ij}(G)}
\end{equation}
with $p_{ij}^*=\frac{e^{-\theta_i^* - \theta_j^*}}{1+e^{-\theta_i^* - \theta_j^*}}$, where
$\vec{\theta}^*$ has to be tuned such that
\begin{equation}
\sum_{j\neq i} p_{ij}^* = k_i^*, \qquad 1\leq i \leq n.
\end{equation}
On the other hand, the solution of \eqref{pbm} via the Lagrange multiplier method gives that
\begin{equation}
q_{ij}^*=\frac{e^{-\phi_i^* - \phi_j^*}}{1+e^{-\phi_i^* - \phi_j^*}},
\end{equation} 
where $\vec{\phi}^*$ has to be tuned such that  
\begin{equation}
\sum_{j\neq i} q_{ij}^*=k_i^*, \qquad1\leq i \leq n.
\end{equation}
This implies that $q_{ij}^*=p_{ij}^*$ for all $1\leq i \neq j \leq n$. Moreover,

\begin{equation}
\begin{aligned}
&\log \Pcan(G^*) + H(p^*)
= \sum_{1 \leq i < j \leq n} g_{ij}(G^*)\log \left(\frac{p_{ij}^*}{1-p_{ij}^*}\right) 
- \sum_{1 \leq i < j \leq n} p_{ij}^*\log \left(\frac{p_{ij}^*}{1-p_{ij}^*}\right)\\
&= -\sum_{1 \leq i < j \leq n} g_{ij}(G^*)(\theta_i^* + \theta_j^*) 
+ \sum_{1 \leq i < j \leq n} p_{ij}^*(\theta_i^* + \theta_j^*)
= \sum_{i=1}^n \theta_i^* \sum_{j \neq i}(p_{ij}^*-g_{ij}(G^*))=0,
\end{aligned}
\end{equation}
where the last equation follows from the fact that 
\begin{equation}
\sum_{j \neq i} g_{ij}(G^*) = \sum_{j \neq i} p_{ij}^*= k_i^*, \qquad 1 \leq i \leq n.
\end{equation}
\end{proof}

%%%%%%%%% APPENDIX B %%%%%%%%%%%%%%%%%%

\section{Appendix}
\label{appC}

We explain the duality between the sparse regime and the ultra-dense regime.

Let $\vec{k}^*= (k_i^*)_{i=1}^n$ be an ultra-dense degree sequence,
\begin{equation}
\label{UD}
\max_{1\leq i \leq n}(n-1-k^*_i) = o(\sqrt{n}),
\end{equation}
and let $\vec{\ell}^*= (\ell_i^*)_{i=1}^n$ be the \emph{dual} degree sequence defined by 
$\ell_i^* = n-1-k_i^*$. Clearly, $\vec{\ell}^*= (\ell_i^*)_{i=1}^n$ is a sparse degree sequence,
\begin{equation}
\label{UDdual}
\max_{1\leq i \leq n} \ell^*_i = o(\sqrt{n}).
\end{equation}

\begin{lemma}
\label{lemmaultradense}
Let $\Pcan$ and $\widehat{\Pcan}$ denote the canonical ensembles in \eqref{eq:PC} 
when $\vC^*=\vec{k}^*= (k_i^*)_{i=1}^n$, respectively, $\vC^*=\vec{\ell}^*= (\ell_i^*)_{i=1}^n$.
Then 
\begin{equation}
\label{dualrel}
\Pcan(G) = \widehat{\Pcan}(G_c), \qquad G\in \cG_n,
\end{equation}
where $G$ and $G_c$ are complementary graphs, i.e.,
\begin{equation}
g_{ij}(G_c) = 1-g_{ij}(G), \qquad 1 \leq i \neq j \leq n. 
\end{equation}
\end{lemma}

\begin{proof}
From the definition of the canonical probabilities we have
\begin{equation}
\Pcan(G)=\Pcan(G\mid \vt^*), \qquad \widehat{\Pcan}(G)=\Pcan(G\mid \vf^*),
\end{equation}
where 
\begin{equation}
\Pcan(G\mid \vt) = \frac{\exp[-\vt \cdot \vec{k}(G)]}{Z(\vt)}, 
\qquad \vec{k}(G)=\sum_{j\neq i} g_{ij}(G),
\end{equation}
and the values $\vt^*$ and $\vf^*$ are such that  
\begin{equation}
\label{free1}
\frac{\partial F(\vt\,)}{\partial \theta_i}\bigg|_{\vt=\vt^*} 
= -\langle k_i \rangle_{\Pcan(\cdot\,\mid\,\vec{\theta}^*)} = -k_i^*,
\end{equation}
\begin{equation}
\label{free2}
\frac{\partial F(\vt\,)}{\partial \theta_i}\bigg|_{\vt=\vf^*} 
= -\langle k_i \rangle_{\Pcan(\cdot\,\mid\,\vec{\phi}^*)} =-\ell_i^*.
\end{equation}
The free energy is $F(\vt)=\log Z(\vt)$, and its $i$-th partial derivative in the Lagrange 
multiplier that fixes the average of the $i$-th constraint. We show that $\vt^*=-\vf^*$.

Write
\begin{equation}
\label{ftheta}
Z(\vt\,) =  \sum_{G\in\cG_n} e^{-\vt\cdot \vec{k}(G)} =
\sum_{G\in\cG_n} e^{-\sum_{i=1}^n \theta_i (n-1-k(G_c))}  = e^{-(n-1)\sum_{i=1}^n \theta_i}\ Z(-\vt\,).
\end{equation}
Using \eqref{free1} and \eqref{ftheta}, we get 
\begin{equation}
-k_i^* = \frac{\partial F(\vt\,)}{\partial \theta_i}\bigg|_{\vt=\vt^*} 
= -(n-1) + \langle k_i \rangle_{\Pcan(\cdot\,\mid\,-\vec{\theta^*})}.
\end{equation}
Since $k_i^* = n-1-\ell_i^*$, we obtain 
\begin{equation}
\label{hi}
\ell_i^* = \langle k_i \rangle_{\Pcan(\cdot\,\mid\,-\vec{\theta}^*)}.
\end{equation}
From \eqref{free2}, \eqref{hi} and the uniqueness of the Lagrange multipliers, we get 
\begin{equation}
\label{lagrmult}
\vt^*=-\vf^*.
\end{equation}
Using \eqref{ftheta} and \eqref{lagrmult}, we obtain
\begin{equation}
\begin{aligned}
\widehat{\Pcan}(G_c)
&=\Pcan(G_c\mid \vf^*)=\Pcan(G_c\mid -\vt^*) 
= \frac{\exp[\vt^* \cdot \vec{k}(G_c)]}{Z(-\vt^*)}\\ 
&= \frac{\exp[-\vt^* \cdot \vec{k}(G)]}{Z(-\vt^*)\,e^{-(n-1)\sum_{i=1}^n\theta^*_i}}  
= \frac{\exp[-\vt^* \cdot \vec{k}(G)]}{Z(\vt^*)}
= \Pcan(G),  
\end{aligned}
\end{equation}
which settles \eqref{dualrel}.
\end{proof}

\begin{lemma}
\label{Ultrarelativeentropy}
Let 
\begin{itemize}
\item
$\Pmic$ and $\Pcan$ denote the microcanonical ensemble in \eqref{eq:PM}, respectively, 
the canonical ensemble in \eqref{eq:PC}, when $\vC^*=\vec{k}^*= (k_i^*)_{i=1}^n$ with 
$k_i^*$ satisfying \eqref{UD}.
\item
$\widehat{\Pmic}$ and $\widehat{\Pcan}$ denote the microcanonical ensemble in \eqref{eq:PM}, 
respectively, the canonical ensemble in \eqref{eq:PC}, when $\vC^*=\vec{\ell}^*= (\ell_i^*)_{i=1}^n$
with $\ell_i^* = n-1-k_i^*$ the dual degree satisfying \eqref{UDdual}. 
\end{itemize}
Then the relative entropy in \eqref{eq:KL2} satisfies 
\begin{equation}
S_n(\Pmic \mid \Pcan) = S_n(\widehat{\Pmic} \mid \widehat{\Pcan}).
\end{equation}
\end{lemma}

\begin{proof}
Consider a graph $G^*$ with degree sequence $\vec{k}(G^*)=\vec{k}^*$. Then 
\begin{equation}
\label{micUD}
\Pmic(G^*) = | \{G \in \cG_n\colon\, \vec{k}(G) = \vec{k}^* \} |^{-1} 
= | \{G \in \cG_n\colon\, \vec{k}(G) = \vec{\ell}^* \} |^{-1} = \widehat{\Pmic}(G^*_c),
\end{equation}
where $G^*_c$ and $G^*$ are complementary graphs, so that $\vec{k}(G^*_c) = \vec{\ell}^*$. 
Using Lemma~\ref{lemmaultradense}, we have
\begin{equation}
\label{canUD}
\Pcan(G^*) = \widehat{\Pcan}(G^*_c).
\end{equation}
Combine \eqref{eq:KL2}, \eqref{micUD} and \eqref{canUD}, to get
\begin{equation}
S_n(\Pmic \mid \Pcan) = \log \frac{\Pmic(G^*)}{\Pcan(G^*)} 
= \log \frac{\widehat{\Pmic}(G^*_c)}{\widehat{\Pcan}(G^*_c)} 
= S_n(\widehat{\Pmic} \mid \widehat{\Pcan}).
\end{equation}
\end{proof}

%%%%%%%%%%%%% BIBLIOGRAPHY %%%%%%%%%%%%%%%%%%%

\end{document}